\documentclass[a4paper,11pt]{article}

\usepackage[utf8]{inputenc}
\usepackage[T1]{fontenc}
\usepackage{microtype}

\usepackage[margin=1in]{geometry}

\usepackage{nicefrac}
\usepackage{amsmath}
\usepackage{amsfonts}
\usepackage{amssymb}
\usepackage{amsthm}
\usepackage{amstext}
\usepackage{thmtools}

\usepackage{csquotes}
\usepackage{xcolor}

\usepackage{mathtools}
\usepackage{xspace}
\usepackage{enumerate}
\usepackage[ruled]{algorithm2e}
\usepackage{paralist}
\usepackage{bm}
\usepackage{dsfont}
\usepackage{thm-restate}
\usepackage{booktabs}
\usepackage{enumitem}
\usepackage{comment}
\usepackage{svg}
\usepackage{complexity}

\usepackage{url}
\usepackage[colorlinks]{hyperref}
\hypersetup{citecolor=green!60!black,linkcolor=blue!90!black,urlcolor=blue!90!black,breaklinks}
\usepackage[capitalize,noabbrev,nameinlink]{cleveref}

\usepackage[skip=5pt plus 3pt minus 1pt,parfill=0pt]{parskip}
\makeatletter
\def\thm@space@setup{%
  \thm@preskip=\parskip \thm@postskip=0pt
}
\makeatother

\usepackage[
 backend=biber,
 style=alphabetic,
 citetracker,
 hyperref=auto,
 maxcitenames=5, %
 sortcites,      %
 sorting=nyt,
 maxbibnames=12, %
 date=year,      %
 isbn=false,     %
 url=false,      %
 doi=false,      %
 eprint=false,   %
]{biblatex}
\newbibmacro{string+doiurlisbn}[1]{%
  \iffieldundef{doi}{%
    \iffieldundef{url}{%
      #1
    }{%
      \href{\thefield{url}}{#1}%
    }%
  }{%
    \href{http://dx.doi.org/\thefield{doi}}{#1}%
  }%
}

\DeclareFieldFormat%
[article,inbook,incollection,inproceedings,patent,thesis,unpublished]
  {title}{\usebibmacro{string+doiurlisbn}{#1}}

\newtheorem{theorem}{Theorem}
\newtheorem{corollary}[theorem]{Corollary}
\newtheorem{lemma}[theorem]{Lemma}

\theoremstyle{definition}

\DeclareMathOperator*{\argmin}{arg\,min}
\newcommand{\bigO}{\mathcal{O}}

\newcommand{\EX}{\mathbb{E}}
\newcommand{\ind}{\ensuremath{\mathds{1}}}
\newcommand{\cp}{C^{\text{LP}}}
\newcommand{\xp}{x^{\text{LP}}}

\newcommand{\LP}{\text{LP}}

\usepackage{xspace}
\newcommand{\alg}{{\normalfont\textsc{Alg}}\xspace}
\newcommand{\opt}{{\normalfont\textsc{Opt}}\xspace}

\newcommand{\ms}{\mathcal{S}}
\newcommand{\mc}{\mathcal{C}}

\newcommand{\abc}[3]{$#1\,|\, #2\,| #3$} %
\newcommand{\pmtn}{\mathrm{pmtn}}

\usepackage[colorinlistoftodos,prependcaption,textsize=scriptsize]{todonotes}

\title{Polytope Scheduling with Groups:\\ Unified Models and Optimal Guarantees}

 \author{Alexander Lindermayr\thanks{Technische Universität Berlin, Germany, \texttt{alexander.lindermayr@tu-berlin.de}. This work was done when the author was affiliated with the University of Bremen and supported by the ``Humans on Mars Initiative'', funded by the Federal State of Bremen and the University of Bremen.} \and
 Zhenwei Liu\thanks{Zhejiang University, Hangzhou, China, \texttt{lzw98@zju.edu.cn}. This work was done when the author was also affiliated with the University of Bremen.} \and
 Nicole Megow\thanks{University of Bremen, Germany, \texttt{nicole.megow@uni-bremen.de}.}}

 \date{}

\begin{document}

\maketitle              %
\begin{abstract}
We propose new abstract and unified perspectives on a range
of scheduling and graph coloring problems with general min-sum 
objectives. Specifically, we consider various problems where 
the objective function is the weighted sum of completion times 
over groups of entities (jobs, vertices, or edges), thereby 
generalizing two important objectives in scheduling: makespan 
and the sum of weighted completion times.

As one of our main results, we present a best-possible 
$\mathcal O(\log g)$-competitive algorithm in the non-clairvoyant 
online setting, where $g$ denotes the size of the largest group.
This is the first non-trivial competitive bound for several problems 
with group completion time objective, and it is an exponential improvement 
over previous results for non-clairvoyant coflow scheduling. For 
offline scheduling, we provide elegant yet powerful meta-frameworks 
that, in a unifying way, yield new or stronger approximation algorithms 
for our new abstract problems as well as for previously well-studied 
special cases.
\end{abstract}

\section{Introduction}

We introduce new unifying abstractions %
of machine scheduling, graph scheduling, and graph coloring problems with general min-sum objectives, and present algorithms for various information settings, both offline models and online models.
More specifically, we focus on problems from these domains where the goal is to minimize the (weighted) sum of certain values $C_j$ for certain entities $j \in J$ (jobs, vertices, or edges), such as the sum of completion times in scheduling, but also its generalization to groups: given a family~$\ms$ of subsets of entities and weights $w_S > 0$ for each $S \in \ms$, we seek to minimize the sum of $\sum_{S \in \ms} w_SC_S$, where $C_S = \max_{j\in S} C_j$. Notably,
the latter objective generalizes both makespan---where %
all jobs belong to a single group---and the sum of weighted completion times---where each job is a group itself.

We establish new connections between well-studied problems and offer a unified perspective that leads to a clearer and more systematic understanding of the problems and their algorithmic solutions.
This perspective allows us to transfer techniques across seemingly very different problems with group completion time objectives. We establish new optimal bounds for non-clairvoyant online scheduling with group completion times, and improve previous approximation ratios for machine scheduling~\cite{DBLP:journals/mor/CorreaSV12} and graph scheduling/coloring problems~\cite{DBLP:conf/swat/DarbouyF24,DBLP:journals/talg/GandhiHKS08}.

Our results cover the following well-studied problem domains; we give more details later. 
\begin{itemize}
  \item \textbf{Machine Scheduling:} In machine scheduling problems, jobs need to be scheduled on parallel machines, where each
  job may have a different processing time on each machine. There are variants with and without preemption and/or migration. Well-studied objective functions are the sum of completion times~\cite{DBLP:journals/mor/HallSSW97,DBLP:journals/jacm/Skutella01,ImKMP14,BansalSS16}, and the sum of group completion times (aka {\em scheduling orders} or {\em bag-of-tasks scheduling})~\cite{leung2006approximation,DBLP:journals/scheduling/YangP05,DBLP:conf/approx/00010021,DBLP:journals/mor/CorreaSV12,DBLP:conf/approx/00010021}.
  \item\textbf{Polytope Scheduling:}
  The polytope scheduling problem generalizes preemptive unrelated machine scheduling~\cite{DBLP:journals/mor/HallSSW97}
  and other preemptive scheduling problems, such as multidimensional scheduling~\cite{DBLP:conf/nsdi/GhodsiZHKSS10} and broadcast scheduling~\cite{DBLP:journals/siamcomp/BansalCS08}. Here, the goal is to select at any time a vector in a packing polytope, which describes the amounts of progress that jobs make;
  see details later. While the sum of weighted completion times objective is well understood~\cite{DBLP:journals/jacm/ImKM18,JLM25,CIP25}, we are not aware of results for group completion times. %
  \item\textbf{Sum Multicoloring:}
   Given a graph with vertex weights $w_v$ and demands~$p_v$, the goal is to assign each vertex $v$ a set of $p_v$ colors (represented by $\mathbb N$), ensuring that no two adjacent vertices share a color. The goal is to minimize the weighted sum of the largest color assigned to a vertex over all vertices~\cite{DBLP:journals/jal/Bar-NoyK98,DBLP:journals/talg/GandhiHKS08,DBLP:journals/jal/HalldorssonK02}. The notion of ``no preemption'' means here that the colors assigned to a vertex must be consecutive.
   The group completion time objective (largest color per group) was studied %
   for unit demands~\cite{DBLP:conf/swat/DarbouyF24}. 
   \item \textbf{Graph Scheduling:}
   Given a graph, the task is to schedule edges with processing times such that %
   the set of scheduled edges at any time forms a matching.
   Common objective functions include
   the sum of edge completion times~\cite{DBLP:journals/talg/GandhiHKS08,DBLP:journals/talg/HalldorssonKS11},
   the sum of grouped edge completion times (called {\em coflow scheduling}~\cite{DBLP:journals/algorithmica/AhmadiKPY20,DBLP:conf/icalp/ImMPP19,DBLP:conf/approx/Fukunaga22,DBLP:conf/wiopt/BhimarajuNV20,RohwedderS25}),
   and the sum of vertex completion times (called {\em data migration}~\cite{DBLP:journals/jal/Kim05,DBLP:journals/talg/GandhiHKS06,DBLP:journals/siamcomp/Mestre10}), where a vertex completes when all its incident edges are completed, %
   a special case of group completion times.
   Graph scheduling on bipartite graphs %
   generalizes the fundamental {\em open shop scheduling} problem~\cite{DBLP:journals/jal/QueyranneS02,QueyranneS02non-preemptive,DBLP:journals/talg/GandhiHKS13}.
\end{itemize}

Our abstract problem formulations are the following.
In the \emph{polytope scheduling problem with group completion times} (PSP-G),
we are given $n$ jobs $j \in J$ with processing requirements $p_j\geq 0$, and a family of groups $\ms \subseteq 2^J$. Each group $S\in \ms$ has associated a positive weight~$w_S>0$.
At any time~$t$, a solution must choose a \emph{rate vector} $y(t)$ in a given downward-closed polytope $\mathcal{P}=\{y\in \mathbb{R}_{\geq 0}^n \mid B\cdot y \leq \mathbf{1}\}$ for some fixed $B = [b_{d,j}]^{D\times n} \in \mathbb{Q}_{\geq 0}^{D\times n}$.
For every job $j$, the entry $y_j(t)$ describes how much processing job $j$ receives at time $t$. Thus, its completion time is $C_j \coloneq \argmin_{t'} \big( \int_0^{t'} y_j(t) \, \mathrm{d}t \geq p_j \big)$.
The objective is to minimize the sum of weighted group completion times $\sum_{S \in \ms} w_S C_S$.
The polytope scheduling problem for the sum of weighted completion times objective (without groups) was first studied by Im, Kulkarni, and Munagala~\cite{DBLP:journals/jacm/ImKM18}, and elegantly abstracts problem-specific details from various preemptive scheduling environments. %
PSP-G captures all the above-mentioned \emph{preemptive} scheduling problems, for which we give formal arguments later. Note that this model does not capture time-dependent constraints like non-preemptiveness or consecutive colorings.

To %
address \emph{non-preemptive} problems,
we introduce a novel discrete variant of PSP-G, which we call \emph{discrete polytope scheduling with group completion times}  (DPSP-G).
On a high-level, it is the discrete analogue to PSP:
We are given (possibly implicitly) a subset of points $\mathcal{P}' \subseteq \mathcal{P}$ of a given downward-closed polytope $\mathcal{P}=\{y\in \mathbb{R}_{\geq 0}^n \mid B\cdot y \leq \mathbf{1}\}$ for some fixed $B = [b_{d,j}]^{D\times n} \in \mathbb{Q}_{\geq 0}^{D\times n}$.
At any time $t$, a schedule must choose a rate vector $y(t) \in \mathcal{P}'$, describing the processing that jobs receive. Additionally, for each job $j$, there is a time $S_j$ such that $y_{j}(t) = y_{j}$ for all $S_j \leq t < C_j$ and $y_{j}(t) = 0$ for all $t < S_j$ and $t \geq C_j$.
The objective is to minimize %
$\sum w_S C_S$.
DPSP-G captures the above-mentioned \emph{non-preemptive} scheduling and coloring problems. %

We study these problems in %
(i) the \emph{non-clairvoyant online} model, where the processing requirements $p_j$ are uncertain, and the algorithmic challenge is of information-theoretic nature,
and (ii) the offline model, where all parameters of an instance are given, and the challenge is the computational complexity. %

\subsection{Our Results}

\begin{table}[tb]
  \caption{An overview of our new and improved results in comparison to previous work.}\label{table:results}
  \setlength{\tabcolsep}{8pt}
  \centering
  \begin{tabular}{lrr}
    \toprule
    problem & known bound & new bound \\
    \midrule
    non-clairvoyant PSP-G & - & $\bigO(\log (\max_{S\in \ms}|S|))$ \\
    non-clairvoyant coflow scheduling & $\bigO(\max_{S\in \ms}|S|)$~\cite{DBLP:conf/wiopt/BhimarajuNV20} & $\bigO(\log (\max_{S\in \ms}|S|))$ \\
    \midrule
    \abc{Q}{r_{j}}{\sum w_SC_S} & $13.5$~\cite{DBLP:journals/mor/CorreaSV12} & $10.874$ \\
    npSMC-G on interval graphs & $11.273$~\cite{DBLP:journals/talg/GandhiHKS08} & $10.874$ \\
    SC-G on perfect graphs & $10.874$~\cite{DBLP:conf/swat/DarbouyF24} & $5.437$ \\
    \midrule
    offline PSP-G & - & $2+\varepsilon$ \\
    \abc{R}{r_{j},\pmtn}{\sum w_SC_S} & $4+\varepsilon$~\cite{DBLP:journals/mor/CorreaSV12} & $2+\varepsilon$ \\
    \bottomrule
  \end{tabular}
\end{table}

We develop {\em general algorithmic frameworks} by framing diverse problems in machine scheduling, graph scheduling, and graph coloring as (discrete) polytope scheduling with groups (PSP-G and DPSP-G).
We give both non-clairvoyant online algorithms and offline approximation algorithms; these are summarized in~\Cref{table:results} with a comparison to prior work.
Abstracting across previously unrelated settings and designing unified algorithms is technically challenging and
constitutes a non-trivial contribution of our work.
In doing so, we even improve substantially upon previous results for certain scheduling and coloring problems.

\medskip
\noindent \textbf{Non-Clairvoyant Online Scheduling.}
In non-clairvoyant scheduling, an algorithm has no knowledge about the processing requirements $p_j$ of jobs until they complete~\cite{DBLP:journals/tcs/MotwaniPT94,DBLP:journals/siamcomp/ShmoysWW95}. Moreover, it must create a schedule over time, and cannot revert decisions from previous times. We say that an online algorithm is $\alpha$-competitive if it computes a solution with an objective value of at most $\alpha \cdot \opt$, where $\opt$ denotes the optimal objective value for the instance. %
\begin{restatable}{theorem}{thmNonClairvoyant}\label{thm:nonClairvoyant}
  There is an $\bigO(\log (\max_{S\in \ms}|S|))$-competitive non-clairvoyant algorithm for PSP with group completion times, PSP-G.
\end{restatable}
This bound is best-possible; we show a lower bound of $\Omega(\log \max_{S \in \ms} |S|)$ for any maximum group size. %
Our result is the first non-trivial bound on the competitive ratio for non-clairvoyant PSP-G and for many of its subproblems, including unrelated machine scheduling with group completion times and data migration. For coflow scheduling, it improves exponentially upon the previous linear competitive ratio in the largest coflow~(i.e., group) size~\cite{DBLP:conf/wiopt/BhimarajuNV20}. %
\begin{corollary}
  There exists a single non-clairvoyant algorithm that is
  \begin{itemize}[nosep]
    \item $\bigO(\log \rho)$-competitive for coflow scheduling, where $\rho$ is the largest coflow size,
    \item $\bigO(\log \Delta)$-competitive for data migration, where $\Delta$ is the largest %
    degree, and
    \item $\bigO(\log (\max_{S\in \ms}|S|))$-competitive for unrelated machine scheduling with group completion times.
  \end{itemize}
\end{corollary}
Our algorithm is based on the Proportional Fairness (PF) algorithm, a natural and simple allocation rule from economics~\cite{nash1950bargaining,kaneko1979nash,JainV10}, which  has recently been proven to be very powerful for PSP with the sum of weighted completion time objective~\cite{DBLP:journals/jacm/ImKM18,JLM25}.
For PSP-G, a %
significant challenge  %
arises from the presence of non-trivial groups.
As fairness in PF is traditionally defined over jobs rather than over groups, we cannot directly apply PF and its analysis. %
We overcome this challenge by introducing virtual weights over jobs
and apply PF to those weights; each group distributes its weight evenly to its jobs. %
Surprisingly, this simple adaptation of PF achieves the best-possible competitive ratio, although the analysis becomes significantly more complicated. To compare these virtual weights with the actual instance in the analysis, we introduce a new factor-revealing LP, which was not used in previous analyses of PF~\cite{DBLP:journals/jacm/ImKM18,JLM25}.

\medskip
\noindent \textbf{Clairvoyant Offline Scheduling.}
Our second set of results addresses the offline setting, where all instance parameters %
are available to the algorithm. Both offline PSP-G and DPSP-G are \APX-hard, even for special cases~\cite{DBLP:journals/jal/QueyranneS02,DBLP:journals/jal/Kim05}. Therefore, we resort to approximation algorithms. We say that an algorithm is an $\alpha$-approximation algorithm if it computes a solution in polynomial time with an objective value of at most $\alpha \cdot \opt$ for every instance.

For the {\em discrete} variant, DPSP-G, we %
present an algorithmic framework that reduces the weighted sum of group completion times objective to the makespan objective.
It is inspired by various grouping and doubling techniques from literature~\cite{DBLP:journals/mor/HallSSW97,QueyranneS02non-preemptive,DBLP:journals/talg/GandhiHKS08,DBLP:conf/swat/DarbouyF24}, but requires new ingredients to deal with general polytopes. Moreover, our framework brings these ideas under a more general umbrella, justifying their previously successful usage.
\begin{restatable}{theorem}{dpspMakespan}\label{thm: framework}
  Given a polynomial-time subroutine that computes for any instance $J'$ a schedule of makespan at most $\rho \cdot \max_{d\in [D]} \sum_{j\in J'}b_{d,j}p_j$,
  there is a polynomial-time $(2\rho e+\varepsilon)$-approximation for any $\varepsilon > 0$ for DPSP-G. %
\end{restatable}
By applying the framework,
we obtain constant-factor approximation algorithms for %
minimizing the total weighted group completion time
on identical machines (\abc{P}{r_{j}}{\sum w_SC_S}), related machines (\abc{Q}{r_{j}}{\sum w_SC_S}), even with release dates, and for non-preemptive %
sum multicoloring on special graphs, which means that the colors assigned to a vertex must be consecutive. %
We use npSMC-G to denote %
non-preemptive sum multicoloring with groups and SC-G
for the special case with unit demands.
\begin{theorem} %
We can apply the algorithmic framework to obtain the following results for minimizing $\sum w_SC_S$:
\begin{itemize}[nosep]
  \item $7.249$-approximation on identical parallel machines, \abc{P}{r_{j}}{\sum w_SC_S},
  \item $10.874$-approximation on related machines, \abc{Q}{r_{j}}{\sum w_SC_S},
  \item $10.874$-approximation for npSMC-G on line graphs and interval graphs, and
  \item $5.437$-approximation for SC-G on perfect graphs. %
\end{itemize}
\end{theorem}

For minimizing the total weighted group completion time on related machines, we improve upon the previous approximation ratio of $13.5$~\cite{DBLP:journals/mor/CorreaSV12}. For npSMC-G on interval graphs, a factor of $11.273$ was previously known~\cite{DBLP:journals/talg/GandhiHKS08},
and for SC-G on perfect graphs, a $10.874$-approximation was known~\cite{DBLP:conf/swat/DarbouyF24}.

Finally, for offline PSP-G, we present a $(2+\varepsilon)$-approximation algorithm.

\begin{restatable}{theorem}{thmPreemptive}\label{thm:preemptive}
  There is a $(2+\varepsilon)$-approximation algorithm for offline PSP-G for any $\varepsilon > 0$, even with non-uniform release dates.
\end{restatable}

This algorithm is based on solving an interval-indexed LP relaxation of \mbox{PSP-G}, and randomly stretching and truncating the LP schedule~\cite{JLM25,DBLP:journals/jal/QueyranneS02}.
For our setting, we carefully combine the main ingredients of these two previous results in a non-trivial way. %
This ratio is best-possible up to a factor of $1+\varepsilon$ as PSP-G contains \abc{1}{\pmtn}{\sum_{S}w_S C_S}, which has been shown to be equivalent to \abc{1}{\mathrm{prec}}{\sum_j w_j C_j}~\cite{DBLP:journals/jal/QueyranneS02} and admits a lower bound of $2-\varepsilon$ under a stronger version of the unique game conjecture~\cite{DBLP:conf/focs/BansalK09}.
Our theorem is the first result for the preemptive data migration problem, and improves over the previously known $(4+\varepsilon)$-approximation for %
\abc{R}{r_{ij},\pmtn}{\sum w_SC_S}~\cite{DBLP:journals/mor/CorreaSV12}.

\begin{corollary}
There is a $(2+\varepsilon)$-approximation for preemptive data migration and preemptive unrelated machine scheduling with group completion times, \abc{R}{r_{ij},\pmtn}{\sum w_SC_S}, for any $\varepsilon > 0$, even with non-uniform release dates.
\end{corollary}

\subsection{Details on Covered Problems and Prior Work}
\label{sec:details-covered-problems}
Our abstract PSP-G captures classical, important, and extensively studied problems within their individual settings. However, despite similarities in some of the underlying ideas, these connections have not been unified. For clarity and to illustrate the breadth of these applications, we list the relevant existing results before demonstrating how our framework %
yields unified algorithmic results.

\paragraph{Machine Scheduling with Group Completion Times.}
This problem, also called \emph{scheduling orders}~\cite{DBLP:journals/mor/CorreaSV12}, generalizes classical scheduling problems  by minimizing weighted group completion times $\sum w_S C_S$. %
It is a special case of DPSP-G~or~PSP-G, depending on whether preemption is allowed, since feasible machine rates can be encoded in a polytope~\cite{DBLP:journals/jacm/ImKM18}.
There are $2$-approximation algorithms for %
identical machines,~\abc{P}{}{\sum w_SC_S}~\cite{leung2006approximation,DBLP:journals/scheduling/YangP05}, a $13.5$-approximation for the non-preemptive problem \abc{R}{r_{ij}}{\sum w_SC_S}~\cite{DBLP:journals/mor/CorreaSV12} and a $(4+\varepsilon)$-approximation for preemptive,  migratory scheduling, \abc{R}{r_{ij}, \pmtn}{\sum w_SC_S}~\cite{DBLP:journals/mor/CorreaSV12}.
Deng et al.~\cite{DBLP:conf/soda/Deng0R23} studied unrelated machines with generalized makespan objectives.
A special case of their model is to minimize the sum of makespans over the groups of machines.

\paragraph{Sum Multicoloring (SMC).}
Given an undirected graph $G = (V, E)$ with vertex weights $w_v$ and demands $p_v$, the goal is to assign $p_v$ distinct positive integers (colors) to each vertex $v$ such that adjacent vertices receive disjoint sets.
Let $C_v$ be the largest color assigned to $v$; the objective is to minimize $\sum_{v \in V} w_v C_v$.
If the assigned colors must be consecutive, the problem is known as non-preemptive sum multicoloring (npSMC). The special case with unit demands ($p_v = 1$) is referred to as SC, for which a $\Omega(n^{1-\varepsilon})$ lower bound holds for general graphs~\cite{DBLP:journals/jal/Bar-NoyHKSS00}.
Hence, most works focus on special graphs~\cite{DBLP:books/tf/18/HalldorssonK18}, e.g., a $\bigO(\log n)$-approximation for npSMC on perfect graphs~\cite{DBLP:journals/jal/Bar-NoyHKSS00,DBLP:journals/talg/GandhiHKS08}, a $7.682$ (resp.\ $11.273$)-approximation on line graphs (resp.\ interval graphs)~\cite{DBLP:journals/talg/GandhiHKS08}.

\paragraph{Graph Scheduling.}
In graph scheduling problems, each edge $e$ of a given graph $G = (V,E)$ represents a job with processing time $p_e$.
At any time, an algorithm can schedule a matching in $G$.
Additionally, edges form groups, and we aim to find a schedule minimizing the sum of weighted group completion times. This model captures several well-studied problems, among them:
\begin{itemize}[nosep]
	\item \emph{Data migration:} Each group is the set of incident edges of a vertex. For non-preemptive data migration, %
		the best known approximation %
    are a $2.618$-approximation~\cite{DBLP:journals/siamcomp/Mestre10} for unit processing times and a $4.96$-approximation~\cite{DBLP:journals/talg/GandhiHKS13} for arbitrary processing times.
	\item \emph{Open shop scheduling} is a special case of graph scheduling and data migration on bipartite graphs~\cite{QueyranneS02non-preemptive,DBLP:journals/talg/GandhiHKS13}. Here, the vertices are either machines or jobs, and each group is the set of incident edges of a job vertex. The best-known preemptive algorithm achieves an approximation factor $2+\varepsilon$~\cite{DBLP:journals/jal/QueyranneS02}.
	\item \emph{Coflow scheduling:} %
	Edges have either unit processing times or arbitrary processing times with preemption allowed at integer time points.
  The best known approximation factor is $4$~\cite{DBLP:journals/algorithmica/AhmadiKPY20,DBLP:conf/approx/Fukunaga22}, with a recent improvement to~$3.415$~\cite{RohwedderS25}. A lower bound of $2-\varepsilon$, unless $\P = \NP$, follows from concurrent open shop~\cite{DBLP:conf/coco/SachdevaS13}. %
	A non-clairvoyant algorithm with competitive ratio linear in the largest coflow (group) size is also known~\cite{DBLP:conf/wiopt/BhimarajuNV20}.
\end{itemize}
\noindent Chowdhury et al.~\cite{DBLP:conf/spaa/ChowdhuryKPYY19} considered generalized coflow scheduling where jobs correspond to paths instead of edges, and obtained a $(2+\varepsilon)$-approximation with arbitrary preemption.
Im et al.~\cite{DBLP:conf/icalp/ImMPP19} give a 2-approximation for a matroid variant, %
where scheduled edges must form an independent set in a matroid; the problem fits into the DPSP-G framework. %

\section{An Optimal Algorithm for Non-Clairvoyant PSP-G} \label{sec: nonclairvoyant}

Our algorithm is inspired by the Proportional Fairness (PF) allocation rule \cite{nash1950bargaining,kaneko1979nash,Mou03,JainV10,PTV22}. %
In a seminal work on PSP,
Im, Kulkarni, and Munagala~\cite{DBLP:journals/jacm/ImKM18} showed that
PF is $O(1)$-competitive for minimizing $\sum w_jC_j$ in PSP.
At any time~$t$, their PF implementation selects
for the set of unfinished jobs $U(t)$
the rate vector $y(t) \in \mathcal P$ of the polytope that maximizes %
$\sum_{j \in U(t)} w_j \log{y_j (t)}$.

We propose a new variant of PF for PSP-G.
For any $S \in \ms$, let $S(t):= S \cap U(t)$ denote the set of unfinished jobs in $S$ at time~$t$.
Further, let $\ms(t):=\{S \in \ms \mid S(t) \neq \emptyset \}$ denote the unfinished groups at time $t$.
At any time $t$, we assign
 to each job a \emph{virtual weight} in a way that each unfinished group $S(t)$ distributes its weight evenly to its unfinished jobs, i.e., $w_{j}(t) := \sum_{S: S \in \ms(t), j \in S}  w_{S}/|S(t)|$.
It follows that $\sum_{j\in U(t)} w_{j}(t) = \sum_{S \in \ms(t)} w_S$.
We then use PF with these virtual weights to compute the rates $y_j(t)$ of the jobs in $U(t)$ as follows:
\begin{alignat}{3}
    \text{max} \quad &\sum_{j \in U(t)} w_{j}(t) \cdot \log{y_j (t)}   \tag{PF} \label{PF}  \\
    \text{s.t.} \quad
    & \sum_{j\in U(t)} b_{d,j}\cdot y_{j}(t) \leq 1 &\quad &\forall d \in [D], t\geq 0 \label{pf:c1}\\
    &y_j(t) \geq 0&\quad& \forall j \in U(t) \label{pf:c2}
\end{alignat}
We refer to~\cite{JLM25,DBLP:journals/jacm/ImKM18} for polynomial-time implementations of PF. Let $y(t)$ denote the optimal solution to~\eqref{PF} at time $t$. %

As in the analysis in~\cite{DBLP:journals/jacm/ImKM18}, our analysis uses the Lagrange multipliers $\{\eta_d \mid d \in [D]\}$ of~\eqref{pf:c1} that correspond to $y(t)$. %
All Lagrange multipliers corresponding to~\eqref{pf:c2} are zero because $y_j(t)>0$ for each $j \in U(t)$. %
They satisfy the following KKT conditions:
\begin{enumerate}[label=(\roman*),nosep]
  \item For each $j \in U(t)$, we have
  $
  \frac{w_j(t)}{y_j(t)} = \sum_{d \in [D]}b_{d,j} \eta_d(t) .
  $
  \item For each $d \in [D]$, we have $\eta_d(t) \geq 0$.
  \item For each $d \in [D]$, we have $\eta_d(t) \cdot (1-\sum_{j\in U(t)}b_{d,j}y_j(t))=0$, which implies that if $\eta_d(t) > 0$ then $\sum_{j\in U(t)}b_{d,j}y_j(t)=1$.
\end{enumerate}
Using these conditions, we can relate the Lagrange multipliers to the weights.
\begin{restatable}{lemma}{lemLagrange}\label{lem:Lagrange}
  At any time $t$, $\sum_{d\in [D]}\eta_d(t) = \sum_{j \in U(t)} w_j(t) = \sum_{S \in \ms(t)} w_S$.
\end{restatable}
\begin{proof}
  The second equality follows directly from the definition of the virtual weights of jobs.
  We prove the first equality. By the KKT conditions,
  we have
  \begin{align*}
  \sum_{j \in U(t)} w_j(t) &=  \sum_{j \in U(t)}y_j(t)\sum_{d \in [D]}b_{d,j}\eta_d(t) =
   \sum_{d\in [D]}\eta_d(t) \sum_{j\in U(t)}b_{d,j}y_{j}(t) = \sum_{d\in [D]}\eta_d(t)\ .
  \end{align*}
  The first equality follows from the first KKT condition and the last equality follows from the third KKT condition.
  Hence we conclude with the lemma.
\end{proof}

To relate PF with an optimal solution to PSP-G, we use the following time-indexed LP relaxation~\eqref{TimeIndexLP}, which builds upon the mean busy time relaxation, similar to previous analyses of PF~\cite{JLM25,DBLP:journals/jacm/ImKM18}.
It incorporates an additional strengthening constraint proposed in the context of open shop scheduling~\cite{DBLP:journals/jal/QueyranneS02}, and has been used for related machine scheduling with groups~\cite{DBLP:conf/approx/00010021}.
The variables $x_{j,t}$ ($x_{S,t}$) indicate which fraction of a job (group) is done during time $(t,t+1]$.
\begin{alignat}{3}
  \text{min} \quad &\sum_{S\in \ms} w_S C_S   \tag{LP($\kappa$)} \label{TimeIndexLP} \\
  \text{s.t.} \quad & \sum_{t \geq 0} x_{S,t} = 1 &\quad &\forall S\in \ms \label{TimeIndexLP:group-finishes} \\
  & \sum_{t' =0}^{t} x_{S,t'} \leq \sum_{t' =0}^{t} \frac{x_{j,t'}}{p_j} &\quad & \forall S \in \ms, j \in S, {t} \geq 0 \label{TimeIndexLP:prefix-bound} \\
  & C_S \geq \sum_{t\geq 0} x_{S,t} \cdot t  &\quad &\forall S \in \ms \label{TimeIndexLP:CS} \\
  & \sum_{j\in J} b_{d,j}\cdot x_{j,t} \leq \frac{1}{\kappa} &\quad &\forall d \in [D], t\geq 0 \label{TimeIndexLP:polytope}\\
   & x_{j,t}, x_{S,t} \geq 0 && \forall j \in J, S \in \ms, t \geq 0 \notag
\end{alignat}
Constraint~\eqref{TimeIndexLP:prefix-bound} is based on precedence-constrained LP relaxations~\cite{DBLP:conf/icalp/0001G0019}, %
ensuring that at
any time $t$, the fraction of a group that has been done up to time $t$ is at most that of any job in the group.
For ease of exposition, %
we restrict %
the solution to run at a machine of speed $\frac{1}{\kappa}$ by strengthening %
the right-hand side of~\eqref{TimeIndexLP:polytope} for some fixed $\kappa \geq 1$. %
Since completion times scale with the machine speed, the optimal value of~\eqref{TimeIndexLP} is at most $\kappa \cdot \opt$. %

We next fit the dual of this LP, which can written as follows. %
\begin{alignat}{3}
  \text{max} \quad &\sum_{S \in \ms} \alpha_S - \sum_{d,t} \beta_{d,t}   \tag{DLP($\kappa$)} \label{dual-program}  \\
  \text{s.t.}  \quad
  & \alpha_S - \sum_{t'\geq t}\sum_{j \in S}\gamma_{j,S,t'} \leq t \cdot w_S &\quad & \forall S \in \ms, t \geq 0 \label{dual:c1} \\
  & \sum_{t'\geq t}\sum_{S:j\in S}\frac{\gamma_{j,S,t'}}{p_j} \leq \kappa \sum_{d\in [D]}b_{d,j}\beta_{d,t} &\quad& \forall j \in J,  t \geq 0 \label{dual:c2} \\
  &  \gamma_{j,S,t} \geq 0 &\quad&\forall S \in \ms, j\in S,t\geq 0 \notag \\
  & \beta_{d,t} \geq 0 && \forall d\in [D], t \geq 0 \notag
\end{alignat}
We set the dual variables such that the algorithm's objective value $\alg$ can be bounded by the objective value of the dual.
Due to the job-to-group coupling constraints~\eqref{TimeIndexLP:prefix-bound} of the LP, %
we need to fit more complex dual variables than in~\cite{JLM25,DBLP:journals/jacm/ImKM18}.
For any time~$t$, let $M(t)$ denote the {\em weighted median} of $\big\{\frac{y_j(t)}{p_j} \mid j \in U(t)\big\}$, where~$j$ has weight $w_j(t)$.
We define a dual assignment as follows:
\begin{itemize}[nosep]
  \item For any $t \geq 0, S \in \ms, j \in S$, if $j \in S(t)$, let $\gamma_{j,S,t} := \frac{w_S}{|S(t)|} \cdot \ind \big[ \frac{y_j(t)}{p_j} \leq M(t) \big]$; otherwise,
  let $\gamma_{j,S,t} := 0$.
  \item We set $\alpha_S := \sum_{t\geq 0} \alpha_{S,t}$, where $\alpha_{S,t} := \sum_{j \in S(t)}\gamma_{j,S,t}$.
  \item We set $\beta_{d,t} := \frac{1}{\kappa} \sum_{t'\geq t} \eta_{d}(t') \cdot M(t')$.
\end{itemize}
Note that the median $M(t)$ satisfies the following two inequalities.
\begin{align}\label{median_right}
&\sum_{j \in U(t)} w_j(t) \cdot \ind\left[\frac{y_j(t)}{p_j} \geq M(t)\right] \geq \frac{1}{2} \sum_{j \in U(t)}w_j(t)  \\
 & \sum_{j \in U(t)} w_j(t) \cdot \ind\left[\frac{y_j(t)}{p_j} \leq M(t)\right] \geq \frac{1}{2} \sum_{j \in U(t)}w_j(t) \label{median_left}
\end{align}
We first show that this dual assignment defines a feasible dual solution. Its proof hinges on the proportional distribution of group weights to jobs in the definition of $\gamma$, which has not been used in \cite{JLM25,DBLP:journals/jacm/ImKM18}.
\begin{restatable}{lemma}{dualfeasibility}
  The assignment $(\alpha_S, \beta_{d,t}, \gamma_{j,S,t})$ is a feasible solution to~\eqref{dual-program}.
\end{restatable}
\begin{proof}
  We first show that~\eqref{dual:c1} is satisfied.
  Fix some group $S$ and $ t \geq 0$. We have
  \[
    \alpha_S - \sum_{t'\geq t}\sum_{j \in S(t')}\gamma_{j,S,t'} = \sum_{t' \geq 0}\sum_{j \in S(t')}\gamma_{j,S,t'} - \sum_{t'\geq t}\sum_{j \in S(t')}\gamma_{j,S,t'} =\sum_{ t' = 0}^{t-1} \sum_{j \in S(t')}\gamma_{j,S,t'} \ .
  \]
  Since $
  \sum_{j \in S(t)}\gamma_{j,S,t} \leq \sum_{j \in S(t)} \frac{w_S}{|S(t)|} = w_S$, the above summation is at most $\sum_{ t' = 0}^{t-1} w_S  = t \cdot w_S$.

  Now we prove that~\eqref{dual:c2} is satisfied.
  Fix some job $j$ and time $t \geq 0$. We have
  \begin{align*}
  \sum_{t'\geq t}\sum_{S:j\in S}\frac{\gamma_{j,S,t'}}{p_j}
  &= \frac{1}{p_j}\sum_{t'\geq t} \sum_{S:j\in S}\left(\frac{w_S}{|S(t')|} \cdot \ind \bigg[\frac{y_{j}(t')}{p_j} \leq M(t') \bigg] \right) \\
  &= \sum_{t'\geq t} \frac{w_{j}(t')}{p_j} \cdot \ind \bigg[\frac{y_{j}(t')}{p_j} \leq M(t') \bigg] \\
  &= \sum_{t'\geq t} \frac{w_{j}(t')}{y_j(t')}\frac{y_j(t')}{p_j} \cdot \ind \bigg[\frac{y_{j}(t')}{p_j} \leq M(t') \bigg]\ .
  \end{align*}
  The last equality is valid since we can only consider time $t'$ when $j$ has not been completed, which implies $ y_j(t')>0$.

  We plug in the first KKT condition $\frac{w_j(t)}{y_j(t)} = \sum_{d \in [D]}b_{d,j} \eta_d(t)$ and rewrite the above summation as
  \begin{align*}
    \sum_{t'\geq t} \frac{y_j(t')}{p_j} \ind \bigg[\frac{y_{j}(t')}{p_j}
    \leq M(t') \bigg] \sum_{d \in [D]}b_{d,j} \eta_d(t')
    &\leq \sum_{t'\geq t}  M(t') \sum_{d \in [D]}b_{d,j} \eta_d(t') \\
    &= \sum_{d \in [D]}b_{d,j} \sum_{t'\geq t} \eta_d(t') M(t') \\
    &\leq \kappa \sum_{d\in [D]}b_{d,j}\beta_{d,t} \ .
  \end{align*}
  Finally, by noting that the assignment has non-negative values of dual variables, we conclude that it is feasible for~\eqref{dual-program}.
  \end{proof}
Next, we relate the dual assignment's objective value to the algorithm's objective value.
Let $H_k:= \sum_{i=1}^k \frac{1}{i}$.
\begin{lemma}\label{lem:dual-value}
  $\sum_{S\in \ms} \alpha_S - \sum_{d,t} \beta_{d,t} \geq  \big( \frac12 - \frac {2H_g}{\kappa} \big) \alg$ where $g:= \max_{S \in \ms} |S|$. %
\end{lemma}
The proof of this lemma can be split into two parts.
We first show 
$\sum_{S\in \ms} \alpha_S \geq \frac12 \alg$. This mainly uses that at any time $t$ 
at least half of the unfinished virtual weight is above the weighted median $M(t)$. 
Thus, we can argue that $\sum_{S \in \mathcal S} \alpha_{S,t} \geq \frac12 \sum_{j \in U(t)} w_j(t) = \frac12 \sum_{S \in \mathcal S(t)} w_S$, and conclude by integrating over time.
\begin{lemma}\label{lem:suma}
    $\sum_{S\in \ms} \alpha_S \geq \frac{1}{2}\alg$.
\end{lemma}

  \begin{proof}
    According to our assignment of the dual variables, we have
    \begin{align*}
  \sum_{S\in \ms} \alpha_S &= \sum_{t\geq 0} \sum_{S\in \ms} \alpha_{S,t} = \sum_{t \geq 0} \sum_{S\in \ms} \sum_{j \in S(t)} \gamma_{j,S,t} \\
  &= \sum_{t \geq 0} \sum_{S\in \ms} \sum_{j \in S(t)} \frac{w_S}{|S(t)|}  \cdot \ind \bigg[\frac{y_j(t)}{p_j} \leq M(t) \bigg] \\
  &= \sum_{t \geq 0} \sum_{j \in U(t)} w_j(t)  \cdot \ind \bigg[\frac{y_j(t)}{p_j} \leq M(t) \bigg] \ .
  \end{align*}
  By the definition of $M(t)$ (cf.~\eqref{median_left}), $\sum_{j \in U(t)}w_j(t)  \cdot \ind \bigg[\frac{y_j(t)}{p_j} \leq M(t) \bigg]$ is at least half of the total weight of all jobs in $U(t)$.
  Hence
  \begin{align*}
  \sum_{t \geq 0} \sum_{j \in U(t)} w_j(t)  \cdot \ind \bigg[\frac{y_j(t)}{p_j} \leq M(t) \bigg]
  &\geq \frac{1}{2}\sum_{t \geq 0} \sum_{j \in U(t)} w_j(t) \\
  &= \frac{1}{2}\sum_{t \geq 0} \sum_{S \in \ms(t)} w_S = \frac{1}{2}\alg \ .
  \end{align*}
  This completes the proof of the lemma.
  \end{proof}

The proof of $\sum_{d,t} \beta_{d,t} \leq \frac{2H_g}{\kappa} \alg$
is particularly challenging and different from~\cite{JLM25,DBLP:journals/jacm/ImKM18}.
The key ingredient is to bound the following expression for any group
$S \in \ms$ and $t\geq 0$:
\begin{equation}
  \sum_{t'\geq t}\sum_{j\in S(t')}  \frac{1}{|S(t')|} \frac{y_{j}(t')}{p_j} \label{eq:key-expression}
\end{equation}
In \cite{DBLP:journals/jacm/ImKM18,JLM25}, where each group is of size 1, this expression is clearly at most $1$, because each job receives at most $p_j$ processing. Using this fact here directly leads to an upper bound of $g$, and $O(g)$-competitiveness.
Yet we prove the exponentially smaller bound of $H_g$ by carefully analyzing how this expression evolves over time.
By discretizing time by job completions $t_1 \leq \ldots \leq t_{|S|}$ of $S$, we can bound \eqref{eq:key-expression} by $\sum_{i =1}^{|S|} \Delta_i$ where $\Delta_i := \sum_{t' = t_{i-1}+1}^{t_{i}} \sum_{j \in S(t')} \frac{1}{|S(t')|} \frac{y_j(t')}{p_j}$. Since the total throughput to any prefix of $\ell$ jobs until time $t_\ell$ must be at least $\ell$, that is, $\sum_{i=1}^\ell (|S|+1-i)\Delta_i \geq \ell$, and the total throughput to all jobs in $S$ can be at most $|S|$,
we can view $\sum_{i =1}^{|S|} \Delta_i$ as the objective of a {\em factor revealing LP} and derive that its maximum, and thus \eqref{eq:key-expression}, is at most $H_{|S|}$.
We formalize this intuition in the following two lemmas, which then imply $\Cref{lem:dual-value}$.

\begin{restatable}{lemma}{lemsumb}
\label{lem:sumb}
  $\sum_{d,t} \beta_{d,t} \leq  \frac {2H_g}{\kappa}  \alg$ where $g:= \max_{S \in \ms} |S|$. %
\end{restatable}

\begin{proof}
Fix some $t$. By \Cref{lem:Lagrange}, we have
\begin{align*}
  \sum_{d\in [D]} \beta_{d,t}=\frac{1}{\kappa} \sum_{d \in [D]} \sum_{t'\geq t} M(t')\eta_{d}(t')
  &= \frac{1}{\kappa} \sum_{t'\geq t} M(t')\sum_{S \in \ms(t')} w_S \\
  &=
  \frac{1}{\kappa} \sum_{t'\geq t} M(t')\sum_{j \in U(t')} w_{j}(t') \ .
\end{align*}
The definition of the weighted median $M(t')$ (cf.~\eqref{median_right}) yields
\[
\sum_{j \in U(t')} w_j(t') \cdot \ind\left[\frac{y_j(t')}{p_j} \geq M(t')\right] \geq \frac{1}{2} \sum_{j \in U(t')}w_j(t')\ .
\]

Thus, the above summation $\frac{1}{\kappa} \sum_{t'\geq t} M(t')\sum_{j \in U(t')} w_{j}(t') $ is at most
\begin{align*}
  &\frac{2}{\kappa} \sum_{t'\geq t} M(t') \sum_{j \in U(t')} w_{j}(t')  \cdot \ind\left[\frac{y_{j}(t')}{p_j} \geq M(t')\right] \\
  &\leq \frac{2}{\kappa} \sum_{t'\geq t} \sum_{j \in U(t')} w_{j}(t') \frac{y_{j}(t')}{p_j} \\
  &= \frac{2}{\kappa} \sum_{S \in \ms(t)} w_{S} \sum_{t'\geq t}\sum_{j\in S(t')}  \frac{1}{|S(t')|} \frac{y_{j}(t')}{p_j} \ .
\end{align*}

Fix some group $S \in \ms$. We claim that for any $t\geq 0$, it holds that
\begin{equation}
  \sum_{t'\geq t}\sum_{j\in S(t')}  \frac{1}{|S(t')|} \frac{y_{j}(t')}{p_j} \leq H_{|S|}\, . \label{eq:nonclairvoyant:claim}
\end{equation}
If true, then it follows %
that
\[
  \frac{2}{\kappa} \sum_{S \in \ms(t)} \sum_{t'\geq t}\sum_{j\in S(t')}  \frac{w_{S}}{|S(t')|} \frac{y_{j}(t')}{p_j} \leq  \frac{2}{\kappa} \sum_{S \in \ms(t)} H_{|S|}\cdot w_S \ ,
\]
and by summing over all $t$, we conclude the proof with
\[
  \sum_{t\geq 0} \sum_{d \in [D]} \beta_{d,t} \leq \frac{2H_g}{\kappa}  \sum_{t \geq 0}\sum_{S \in \ms(t)} w_S = \frac{2H_g}{\kappa} \alg \ .
\]

It remains to prove~\eqref{eq:nonclairvoyant:claim}.
Note that for any fixed $S \in \ms$, the left-hand-side summation
\[\sum_{t'\geq t}\sum_{j\in S(t')}  \frac{1}{|S(t')|} \frac{y_{j}(t')}{p_j}
\]
is non-increasing in $t$.
Hence, it suffices to show that the claim holds when $t=0$.

Now we prove the claim for $t=0$. Let $r:=|S|$ and $j_1, \dots, j_r \in S$ be the jobs of group $S$, indexed in non-decreasing order of completion time.
We break the summation by whenever a job is completed and $|S(t')|$ decreases by~$1$.
Let~$t_i$ be the completion time of $j_i$ ($t_0=-1$) and $\Delta_i:=\sum_{t_{i-1}< t' \leq t_i}\sum_{j\in S(t')}  \frac{1}{|S(t')|} \frac{y_{j}(t')}{p_j}$.
Using this notation,
we want to prove $\sum_{i=1}^r \Delta_i \leq H_r$.
We show that $\{\Delta_i\}$ satisfy certain constraints.
First, as the total fractions done over all jobs is at most $r$, we have $\sum_{i=1}^r (r+1-i)\Delta_i \leq r$.
Second, for any $\ell \leq r$, the total fractions done of the first $\ell$ jobs before $j_{\ell}$ is completed, which is
\[
\sum_{t' \leq t_\ell} \sum_{j \leq \ell} \frac{y_j(t')}{p_j} \leq \sum_{t' \leq t_\ell} \sum_{j \in S(t')} \frac{y_j(t')}{p_j} = \sum_{i=1}^\ell (r+1-i)\Delta_i \ ,
\]
must be at least $\ell$.
Thus, $\sum_{i=1}^r \Delta_i$ is bounded by the following~\eqref{factorlp}.
\begin{alignat}{3}
  \text{max} \quad &\sum_{i=1}^r \Delta_i  \tag{FactorLP} \label{factorlp} \\
  \text{s.t.} \quad
  & \sum_{i=1}^r (r+1-i)\Delta_i \leq r &\quad &  \notag \\
  & \sum_{i=1}^\ell (r+1-i)\Delta_i \geq \ell &\quad& \forall 1 \leq \ell \leq r \notag\\
  &  \Delta_i \geq 0 &\quad&\forall 1 \leq i \leq r \notag
\end{alignat}
We show in the following that $\Delta_i = \frac{1}{r+1-i}$ for $1\leq i\leq r$ is an optimal solution to~\eqref{factorlp}. Since $\sum_{i=1}^r \Delta_i = \sum_{i=1}^r \frac{1}{r+1-i} = H_r$, we conclude with Claim~\eqref{eq:nonclairvoyant:claim} and %
the lemma.

The dual of~\eqref{factorlp} can be written as follows:
\begin{alignat}{3}
  \text{min} \quad &r \cdot a - \sum_{\ell=1}^r \ell \cdot b_\ell  \notag   \\
  \text{s.t.} \quad
  & (r+1-i)a - \sum_{\ell=i}^r (r+1-i) b_\ell \geq 1 &\quad & \forall 1 \leq i \leq r  \notag \\
  & a,b_\ell \geq 0 &\quad&\forall 1 \leq \ell \leq r \notag
\end{alignat}
Consider the dual assignment $a=1$, $b_r=0$ and $b_\ell = \frac{1}{r-\ell} - \frac{1}{r-\ell+1}$ for $1 \leq \ell \leq r-1$. This solution satisfies each dual constraint  $1\leq i \leq r$ with equality
\begin{align*}
  (r+1-i)a - \sum_{\ell=i}^r (r+1-i) b_\ell
  &= r+1-i + (r+1-i)\sum_{\ell=i}^{r-1} \bigg(\frac{1}{r-\ell+1}-\frac{1}{r-\ell} \bigg) \\
  &= r+1-i + (r+1-i) \bigg(\frac{1}{r-i+1}-1 \bigg) = 1 \ .
\end{align*}
Moreover, the primal solution $\Delta_i$ satisfies each primal constraint $1 \leq \ell \leq r$ with equality
\[
\sum_{i=1}^\ell (r+1-i)\Delta_i = \sum_{i=1}^\ell (r+1-i)\frac{1}{r+1-i} = \ell \ .
\]
Thus, the primal and dual solutions are both optimal by complementary slackness.  
\end{proof}

Having \Cref{lem:dual-value},
we set the parameter $\kappa := 8H_g$ where $g:= \max_{S \in \ms} |S|$ so that the objective value of the dual program is $\sum_{S \in \ms} \alpha_S - \sum_{d,t} \beta_{d,t} \geq \alg/4$.
Since the objective value of any feasible solution is at most $\kappa \opt$,
we have $\alg \leq 4(\sum_{S \in \ms} \alpha_S - \sum_{d,t} \beta_{d,t}) \leq 32H_g \cdot \opt = \bigO(\log g ) \cdot \opt$.
We conclude that our algorithm is $\bigO(\log g)$-competitive.

\thmNonClairvoyant*

\paragraph{Matching lower bound}
We prove a matching lower bound 
that builds on previous work for non-clairvoyant makespan minimization on related machines~\cite{DBLP:conf/approx/00010021,DBLP:journals/siamcomp/ShmoysWW95}.

\begin{restatable}{theorem}{ThmLBNonclc}
\label{thm:lb-nonclv}
  Every non-clairvoyant algorithm for PSP-G has a competitive ratio of at least $\Omega(\log \max_{S \in \ms} |S|)$ for any value of $\max_{S \in \ms} |S|$.
\end{restatable}

\begin{proof}
Shmoys, Wein, and Williamson \cite{DBLP:journals/siamcomp/ShmoysWW95} have shown a lower bound of $\Omega(\log n)$ for non-clairvoyant makespan minimization on related machines. Note that this is a special case of PSP-G.
Fix some parameter $k$. The instance in \cite{DBLP:journals/siamcomp/ShmoysWW95} involves $k$ sets of machines $T_i$ for $1 \leq i\leq k$, each of which contains machines of speed $2^i$, and $k$ sets of jobs $J_i$ for $1 \leq i\leq k$, each of which contains jobs of size $2^i$. Further, $|T_i| = |J_i| = 2^{2k-2i-1}$ for $1 \leq i < k$ and $|T_k|=|J_k|=1$.
Observe that $\opt = 1$ by assigning each job in $J_i$ to a distinct machine in $T_i$ for $1 \leq i \leq k$.
Fix any non-clairvoyant algorithm $\alg$ and let $C_i$ be the completion time of the last completed job in $J_i$ using $\alg$.
In \cite{DBLP:journals/siamcomp/ShmoysWW95} it is shown that 
the adversary can force $C_1 \leq \dots \leq C_k$ and $C_i - C_{i-1} \geq \frac{1}{4}$ for $1 \leq i<k$.
Hence, the makespan of the schedule by $\alg$ is $C_k = \Omega(\log n)$.

We adapt the instances and derive a refined lower bound of $\Omega(\log \max_{S \in \ms} |S|)$ with multiple groups of smaller sizes by the following arguments.
Let $I$ be an instance constructed as above.
To obtain a new instance for PSP-G, we view all the jobs in $I$ as a single group of weight $1$ and introduce more dummy jobs.
Fix some sufficiently small constant $\varepsilon > 0$, we add $2^k/\varepsilon$ dummy groups of weight $\varepsilon^2$ and consisting of a single job with a small processing time $\varepsilon^2$.
We construct a schedule as follows.
First assign all the dummy jobs to the machine with speed $2^k$ and complete all of them before time $\varepsilon$, which contributes at most $\varepsilon \cdot 2^k$ to the objective value.
Then complete the jobs in $I$ as in the optimal schedule of $I$ before time $1+\varepsilon$.
Hence the optimal objective value is $\bigO(1)$ and the adversary can force any non-clairvoyant algorithm to have an objective value of $\Omega(\log \max_{S \in \ms} |S|)$.
In this way, we construct an instance with multiple groups, where the maximum group size can be much smaller than $n$.
\end{proof}

\paragraph{Incorporating release times without performance loss.} Our algorithm can be easily extended to the online-time setting with job release times.
We distinguish two models. In the first, all jobs of a group are released simultaneously, as considered previously in coflow scheduling~\cite{DBLP:conf/approx/Fukunaga22,DBLP:journals/algorithmica/AhmadiKPY20}. In the second, jobs may arrive over time, but the group information (its weight and its group membership) is known in advance. In both models, we can obtain the same competitive ratio by solving~\eqref{PF} with respect to those unfinished but released jobs.

\section{Discrete PSP}\label{sec: DPSP}
We next present our results for DPSP-G. We introduce a general framework to derive polynomial-time approximation algorithms, and then apply it to various subproblems in~\Cref{sec:dpsp-machine-scheduling,sec:dpsp-graph-scheduling} by specifying the remaining ingredients.

Our framework is based on the LP relaxation $(\hyperref[TimeIndexLP]{\LP(1)})$.
As explained in \Cref{sec: nonclairvoyant}, it relaxes PSP-G given an infinitesimal time discretization. Since every feasible schedule for DPSP-G with configurations $\mathcal{P}'$ is feasible for PSP-G with polytope $\mathcal{P} \supseteq \mathcal{P}'$, $(\hyperref[TimeIndexLP]{\LP(1)})$ is also a relaxation of DPSP-G.
The size of $(\hyperref[TimeIndexLP]{\LP(1)})$ can be exponential, %
hence we first rewrite it into a polynomial-size interval form. %
To this end, consider a large enough time horizon~$T$, %
and fix constants $\delta>0$ and $\varepsilon' >0$.
First, we slightly modify the instance by shifting the release time $r_j':= r_j +\delta$ for all $j$.
This increases the optimal objective value by at most a factor of $1+\delta$ 
assuming all jobs complete after time $1$, which can be achieved w.l.o.g.\ by scaling.
Let $L$ be the smallest integer such that $\delta(1 + \varepsilon')^L \geq T$. Note that $L$ is polynomial in the size of the instance.
For every $0 \leq i \leq L$, we define $ \gamma_i := \delta(1 + \varepsilon')^i$, and for every $1 \leq i \leq L$, we define the interval $I_i = (\gamma_{i-1}, \gamma_i]$.
Let $|I_i|$ be the length of the $i$th interval, i.e., $|I_i|:=\gamma_{i}-\gamma_{i-1}$.
We consider the following \emph{interval LP relaxation} of $(\hyperref[TimeIndexLP]{\LP(1)})$, which is analogous to $(\hyperref[TimeIndexLP]{\LP(1)})$ if we replace a time slot $(t,t+1]$ by the interval $(\gamma_{i-1},\gamma_i]$.
\begin{alignat}{3}
  \text{min} \quad &\sum_{S\in \ms} w_S C_S   \tag{LP'} \label{IntervalLP} \\
  \text{s.t.} \quad & \sum_{i=1}^L x_{S,i} \geq 1 &\quad &\forall S\in \ms \label{IntervalLP:group-finishes} \\
  & \sum_{i' =1}^{i} x_{S,i'} \leq \sum_{i' =1}^{i} \frac{x_{j,i'}}{p_j}|I_{i'}| &\quad & \forall S \in \ms, j \in S, 1 \leq i \leq L \label{IntervalLP:prefix-bound} \\
  & C_S = \sum_{i=1}^L x_{S,i} \cdot  \gamma_{i-1}  &\quad &\forall S \in \ms \label{IntervalLP:CS} \\
  & \sum_{j\in J} b_{d,j}\cdot x_{j,i} \leq 1 &\quad &\forall d \in [D], 1 \leq i \leq L \label{IntervalLP:polytope}\\
  & x_{j,i}=0 && \forall j \in J, 1 \leq i\leq L, r_j > \gamma_{i-1} \notag\\
   & x_{j,i}, x_{S,i} \geq 0 && \forall j \in J, S \in \ms, i \geq 0 \notag
\end{alignat}
Here, $x_{j,i}$ (resp.\ $x_{S,i}$) denotes the progress of job $j$ (resp.\ group $S$) during the time interval $I_i$.
Let $\{\xp_{j,i},\xp_{S,i}, \cp_S\}$ be an optimal solution to~\eqref{IntervalLP}.
Then, $\sum_{S \in \ms}w_S\cp_S \leq (1+\delta)\opt$.
Our three-step framework is as follows:
\begin{enumerate}[label=(\roman*)]

  \item Solve~\eqref{IntervalLP} and obtain an optimal solution $\{\xp_{j,i},\xp_{S,i}, \cp_S\}$.
  \item Create batches of jobs according to that solution as follows. First, for every job $j$, define $\cp_j \coloneq \sum_{i=1}^L \xp_{j,i}  |I_i|  \gamma_{i-1} / p_j$. Let $\alpha$ be a variable uniformly drawn from $[0,1]$ and $\beta$ be some parameter to be determined later.
  We define a partition of jobs $J_i \coloneq \{j \in J \mid \beta^{i-1+\alpha}<\cp_j\leq \beta^{i+\alpha}\}$ for $i=0,1,\ldots, K$. %
  \item We schedule the batches $J_0, \ldots, J_K$ sequentially.
  That is, we start scheduling $J_{i+1}$ after all jobs in $J_i$ are completed.
  For each batch, we use a problem-specific scheduling subroutine for the following subproblem:
	
  \medskip 
  \textbf{Subproblem:}
  Given a subset $J'$ of jobs, compute a feasible non-preemptive schedule that completes all the jobs in $J'$ within a makespan of at most
  \(
  \rho \cdot \max_{d\in [D]} \sum_{j\in J'}b_{d,j} \cdot p_j
  \)
  for some $\rho \geq 1$.

\end{enumerate}
Similar batching and makespan reductions have been previously used for more specific scheduling problems~\cite{QueyranneS02non-preemptive,DBLP:journals/talg/GandhiHKS08}.
We unify these approaches by comparing the makespan to the abstract lower bound $\max_{d\in [D]} \sum_{j\in J'}b_{d,j} \cdot p_j$ that holds for all DPSP-G instances.
We now explain the intuition behind the framework.

We can interpret the abstract lower bound as follows.
Consider any subset~$J'$ of job and let $d \in [D]$. Let $y_{j,t}$ be the amount of processing of job $j$ receives at time $t$.
Suppose we complete $J'$ by time $T$.
Since each job $j \in J'$ has received $p_j$ units of processing, we have
\(
\int_{0}^T \sum_{j \in J'}b_{d,j}y_{j,t}  \mathrm{d}t = \sum_{j \in J'} b_{d,j} \int_{0}^T y_{j,t} \mathrm{d}t \geq  \sum_{j \in J'} b_{d,j} p_j .
\)
By the polytope constraints, we have $\int_{0}^T \sum_{j \in J'}b_{d,j}y_{j,t}\mathrm{d}t \leq \int_{0}^T 1 \mathrm{d}t = T$.
Hence, $T$ must be at least $\sum_{j\in J'}b_{d,j} \cdot p_j$.

We next give an outline of the analysis of the framework.
We first show that our time-indexed formulation implies inequalities $(1+\varepsilon')\sum_{j \in J'} b_{d,j} p_j \cp_j \geq \frac{1}{2}(\sum_{j \in J'} b_{d,j}p_j)^2$ for all $J' \subseteq J$.
These can be thought of generalizations of Queyranne's subset constraints for single-machine scheduling~\cite{Queyranne93,Schulz96}, abstracting away the machine environment via the polytope.

\begin{restatable}{lemma}{setbound} \label{lem: set_bound}
For any $J' \subseteq J$ and $d \in [D]$, it holds that
\[
(1+\varepsilon')\sum_{j \in J'} b_{d,j} p_j \cp_j \geq \frac{1}{2}\bigg(\sum_{j \in J'} b_{d,j}p_j \bigg)^2 \ .
\]
\end{restatable}

\begin{proof}
  Recall that $\gamma_{i} = \delta (1+\varepsilon')^i$
  and $|I_i| =\gamma_i -\gamma_{i-1} = \delta \cdot \varepsilon'\cdot(1+\varepsilon')^{i-1}$.
  Observe that
  \begin{align*}
    \sum_{j \in J'} b_{d,j} p_j \cp_j  = \sum_{j \in J'} b_{d,j} p_j \sum_{i=1}^L \frac{\xp_{j,i}}{p_j}|I_i| \gamma_{i-1}
    = \sum_{i=1}^L |I_i| \gamma_{i-1} \sum_{j \in J'}b_{d,j} \xp_{j,i}  \ .
  \end{align*}
  Let $Y:= \sum_{j \in J'} b_{d,j} p_j $ and let $s_i:= |I_i|\sum_{j \in J'}b_{d,j}\xp_{j,i}$.
  Thus, it remains to prove that
  \[
     (1+\varepsilon' )\sum_{i=1}^L s_i \gamma_{i-1} \geq \frac{Y^2}{2} \ .
  \]
  By~\eqref{IntervalLP:polytope}, we have
  $\sum_{j \in J'}b_{d,j} \xp_{j,i} \leq 1$.
  Hence $0 \leq s_i \leq |I_i|$.
  Further, we have
  \[
    \sum_{i=1}^L s_i = \sum_{i=1}^L |I_i|\sum_{j \in J'}b_{d,j}\xp_{j,i}
    =\sum_{j \in J'} b_{d,j} \sum_{i=1}^L |I_i|\xp_{j,i}
    \geq \sum_{j \in J'} b_{d,j} p_j = Y
  \]
  because $\sum_{i=1}^L \xp_{j,i} |I_i| \geq  p_j$, which is implied by \eqref{IntervalLP:group-finishes} and \eqref{IntervalLP:prefix-bound}.

  Let $k$ be the index such that $\sum_{i=1}^k |I_i| \leq Y$ and $\sum_{i=1}^{k+1} |I_i| > Y$.
  Observe that $\sum_{i=1}^L s_i \gamma_{i-1}$ reaches its minimum when $s_i = |I_i|$ for $i\leq k$ and $s_{k+1} = Y-\sum_{i=1}^L s_i < |I_{k+1}|$. Let $r:= s_{k+1}$.
  We have
  \begin{align*}
    \sum_{i=1}^L s_i \gamma_{i-1} &\geq r \cdot \gamma_{k} + \sum_{i=1}^k |I_i|\gamma_{i-1}
    = r \cdot \delta(1+\varepsilon')^k + \sum_{i=1}^k \varepsilon' (1+\varepsilon')^{2i-2}
    \\
    &= r \cdot \delta\cdot (1+\varepsilon')^k + \varepsilon' \delta^2\sum_{i=1}^k (1+\varepsilon')^{2i-2}
    = r \cdot \delta\cdot (1+\varepsilon')^k + \delta^2 \frac{(1+\varepsilon')^{2k}-1}{2+\varepsilon'} \ .
  \end{align*}
  Multiplying $1+\varepsilon'$ and using the fact that $\frac{1+\varepsilon'}{2+\varepsilon'} \geq \frac{1}{2}$, we obtain
  \[
    (1+\varepsilon')\sum_{i=1}^L s_i \gamma_{i-1} \geq \delta \cdot r \cdot (1+\varepsilon')^{k+1} + \delta^2 \frac{(1+\varepsilon')^{2k}-1}{2} \ .
  \]
  Now we consider $Y^2$, which is equal to twice the RHS of the claimed inequality.
  \begin{align*}
    Y^2
    &=  \bigg(r +\sum_{i=1}^k |I_i| \bigg)^2
    =  \left(r + \delta ((1+\varepsilon')^k-1) \right)^2 \\
    &= r^2+ 2 r \delta \left((1+\varepsilon')^k-1\right) + \delta^2 \left((1+\varepsilon')^k-1\right)^2
    \\
    &= r^2 + 2\delta\left((1+\varepsilon')^k-1\right) r + \delta^2(1+\varepsilon')^{2k}-2\delta^2(1+\varepsilon')^k+\delta^2\ .
  \end{align*}
  We now consider $2 \cdot (\text{RHS}-\text{LHS})$:
  \begin{align*}
    &Y^2 - 2(1+\varepsilon')\sum_{i=1}^L s_i \gamma_{i-1} \\
    &\leq  r^2+ 2\delta\left((1+\varepsilon')^k-1-(1+\varepsilon')^{k+1}\right)r -2\delta^2(1+\varepsilon')^k+2\delta^2 \\
    &\leq   r^2- 2\delta \varepsilon' (1+\varepsilon')^k r -2\delta^2(1+\varepsilon')^k+2\delta^2
    \ .
  \end{align*}
  Consider the last summation as a quadratic function of $r$. It is decreasing whenever $r \in [0, \delta \varepsilon' (1+\varepsilon')^k ] \leq |I_{k+1}|$.
  Thus, if $r \leq |I_{k+1}|$, it reaches its maximum at $r=0$, which is $2\delta^2 (1-(1+\varepsilon')^k) < 0$.
  This implies $Y^2 - 2(1+\varepsilon')\sum_{i=1}^L s_i \gamma_{i-1} < 0$, and thus, the lemma.
  \end{proof}

Using \Cref{lem: set_bound}, we can give a bound on the makespan of batch $J_i$.

\begin{lemma} \label{lem: batch_makespan}
For each $i \in \{0,\ldots,K\}$, if we start scheduling batch $i$ at time $0$
with the subroutine, we complete all %
jobs in $J_i$ by time $2 (1+\varepsilon')\rho \beta^{i+\alpha}$.
\end{lemma}
\begin{proof}
  By \Cref{lem: set_bound}, for any $d \in [D]$ we have
\[
\frac{1}{2}\bigg(\sum_{j \in J_i} b_{d,j}p_j \bigg)^2 \leq (1+\varepsilon') \sum_{j \in J_i} b_{d,j} p_j \cp_j \leq (1+\varepsilon')\beta^{i+\alpha} \sum_{j \in J_i} b_{d,j} p_j \ ,
\]
which implies $\sum_{j \in J_i} b_{d,j} p_j \leq 2 (1+\varepsilon') \beta^{i+\alpha}$.
Hence, the subroutine will output a schedule that completes all the job in $J_i$ within time
\[
2(1+\varepsilon')\rho \cdot \max_{d\in [D]} \sum_{j\in J_i}b_{d,j}p_j \leq 2(1+\varepsilon')\rho \beta^{i+\alpha} \ ,
\]
which concludes with the lemma.
\end{proof}

Since we schedule the batches one after another, to bound the completion time of a single job $j\in J_i$, we need to sum up the completion times of $J_0, \ldots, J_{i-1}$.

\begin{lemma}
  For each $i \in \{0,\ldots,K\}$ and for each job $j \in J_i$, we have $C_j \leq 2(1+\varepsilon')\rho \frac{\beta^{i+1+\alpha}}{\beta-1}$.
\end{lemma}
\begin{proof}
We bound $C_j$ by the sum of completion time of the batches $J_0, \dots, J_i$. Hence,
\[
    C_j\leq \sum_{j=0}^i 2 (1+\varepsilon')\rho \beta^{j+\alpha} = 2 (1+\varepsilon')\rho\beta^\alpha \frac{\beta^{i+1}-1}{\beta-1} \leq 2 (1+\varepsilon')\rho\frac{\beta^{i+1+\alpha}}{\beta-1} \ ,
\]
which completes the proof of the lemma.
\end{proof}

\begin{restatable}{lemma}{jobLPcompletion}
  For any $S \in \ms$ and any $j \in S$, it holds that $\cp_j \leq \cp_S$.
\end{restatable}

\begin{proof}
  We rewrite $\cp_S$ as follows.
  \begin{align*}
    \cp_S &= \sum_{i=1}^L \xp_{S,i} \gamma_{i-1}  =
    \sum_{i=1}^L \xp_{S,i} \int_0^{\gamma_{i-1}} 1 \; \mathrm{d}t
    = \int_0^{\infty} \sum_{1 \leq i \leq L, \gamma_{i-1} \geq t} \xp_{S,i} \;\mathrm{d}t \ .
  \end{align*}
  Using \eqref{IntervalLP:group-finishes}, the above is equal to
  \begin{align*}
    &= \int_0^{\infty}  \bigg(1-\sum_{1 \leq i \leq L, \gamma_{i-1} \leq t} \xp_{S,i} \bigg) \;\mathrm{d}t
    \geq \int_0^{\infty} \bigg(1-\sum_{1 \leq i \leq L, \gamma_{i-1} \leq t} \frac{\xp_{j,i}}{p_j}|I_i| \bigg)\;\mathrm{d}t   \\
    &= \int_0^{\infty} \sum_{1 \leq i \leq L, \gamma_{i-1} \geq t} \frac{\xp_{j,i}}{p_j}|I_i| \;\mathrm{d}t
    = \sum_{i=1}^L \frac{\xp_{j,i}}{p_j}|I_i| \int_0^{\gamma_{i-1}} 1 \; \mathrm{d}t \\
    &= \sum_{i=1}^L \frac{\xp_{j,i}}{p_j}|I_i| \gamma_{i-1} = \cp_j
    \ .
  \end{align*}
  The inequality
   follows from \eqref{IntervalLP:prefix-bound}.
\end{proof}

Fix any $S \in \ms$ with $\beta^{i-1+\alpha}<\cp_S\leq \beta^{i+\alpha}$ for some $i$.
For any job $j \in S$, since $\cp_j \leq \cp_S \leq \beta^{i+\alpha}$, job $j$ belongs to some batch $\ell \leq i$.
Hence, $C_S = \max_{j \in S} C_j \leq 2(1+\varepsilon')\rho\frac{\beta^{i+1+\alpha}}{\beta-1}$.

\begin{corollary} \label{cor: cs}
  For any $i \geq 0$ and for any group $S$ with $\beta^{i-1+\alpha}<\cp_S\leq \beta^{i+\alpha}$, the completion time of $S$, $C_S$, is at most $2(1+\varepsilon')\rho\frac{\beta^{i+1+\alpha}}{\beta-1}$.
\end{corollary}

Recall that $\alpha$ is a random variable.
Next, we bound the expected performance of the~algorithm.

\begin{restatable}{lemma}{expectedCS}
  For any group $S \in \ms$, $\EX[C_S] \leq \frac{2(1+\varepsilon')\rho\beta}{\ln \beta} \cp_S$.
\end{restatable}

\begin{proof}
  We rewrite $\cp_S = \beta^{i-1+\gamma}$ for some $i \geq 0, \gamma \in (0,1]$.
  By \Cref{cor: cs},
  when $\alpha < \gamma$, we have $\beta^{i-1+\alpha}<\cp_S\leq \beta^{i+\alpha}$.
  Hence,
  \[
    \frac{C_S}{\cp_S} \leq 2(1+\varepsilon')\rho\frac{\beta^{i+1+\alpha}}{\beta-1} \cdot \frac{1}{\beta^{i-1+\gamma}} = 2(1+\varepsilon')\rho\frac{\beta^{2+\alpha-\gamma}}{\beta-1} \ .
  \]
When $\gamma \leq \alpha < 1$, we have  $\beta^{(i-1)-1+\alpha}<\cp_S\leq \beta^{(i-1)+\alpha}$.
Hence,
\[
  \frac{C_S}{\cp_S} \leq 2(1+\varepsilon')\rho\frac{\beta^{i+\alpha}}{\beta-1} \cdot \frac{1}{\beta^{i-1+\gamma}} = 2(1+\varepsilon')\rho\frac{\beta^{1+\alpha-\gamma}}{\beta-1} \ .
\]
In the first case, $2+\alpha-\gamma \in [2-\gamma, 2]$ and in the second $1+\alpha-\gamma \in [1, 2-\gamma]$.
Thus, we can
rewrite the bound as $\frac{C_s}{\cp_S} \leq 2(1+\varepsilon')\rho\frac{1+x}{\beta-1}$, where $x$ is a uniform random variable in $[0,1]$. Hence, we have
\[
\EX \bigg[\frac{C_S}{\cp_S} \bigg]
\leq 2(1+\varepsilon')\rho\int_0^1
 \frac{\beta^{1+x}}{\beta-1} \; \mathrm{d}x
 = \frac{2(1+\varepsilon')\rho\beta}{\ln \beta},
\]
which implies the lemma.
\end{proof}

By optimizing the choice of $\beta$ (setting $\beta:= e$), we
can bound the approximation ratio with $2e(1+\varepsilon')$ against the optimal solution of~\eqref{IntervalLP}, and thus, an overall approximation ratio of $2e(1+\varepsilon')(1+\delta) = 2e + \varepsilon$ for suitable $\varepsilon'$ and $\delta$.

Above we only argued about the case where all jobs are available at time $0$.
If jobs have release times $r_j$, we need some modification since we cannot start scheduling a job before its release time.
First, we add the constraints $x_{j,t} = 0$ for each job $j$ and $t < r_j$ to the LP, which implies $r_j \leq \cp_j $.
Second, we force the completion time of the batch $J_i$ to be exactly the upper bound $2\rho(1+\varepsilon') \beta^{i+\alpha}$ (cf.\ \Cref{lem: batch_makespan}) by adding idle time.
Hence we only start the batch $J_i$ after time $\sum_{\ell:0\leq \ell<i}  2\rho(1+\varepsilon') \beta^{\ell+\alpha}\leq 2\rho(1+\varepsilon')\frac{\beta^{i+\alpha}}{\beta-1}$.
Since job $j$ belongs to batch $J_i$, we have $r_j \leq \cp_j \leq \beta^{i+\alpha}$.
Therefore, as long as $r_j \leq \beta^{i+\alpha} \leq  2\rho(1+\varepsilon')\frac{\beta^{i+\alpha}}{\beta-1}$, or equivalently $\beta \leq 2\rho(1+\varepsilon') + 1$, we would schedule $j$ only after its release time $r_j$, which always holds when $\beta=e$.
We conclude with the following.
\dpspMakespan*

\subsection{Applications to %
Machine Scheduling} \label{sec:dpsp-machine-scheduling}

We show how to apply our framework to parallel-machine scheduling with group completion times. 
While the identical-machine case mainly serves to illustrate our framework, we improve upon previous results for related parallel machines.
\medskip 
\noindent \textbf{Identical Machines.}
We are given $m$ parallel identical machines, each with a speed of~$1$.
At any time $t$, a schedule must satisfy $y_{j,t} \leq 1$ for any job $j$ and
$\sum_{j}\frac{y_{j,t}}{m} \leq 1$.
These constraints form a feasible polytope for PSP. %
Thus, %
we have
\(
\max_{d\in [D]} \sum_{j\in J'}b_{d,j}p_j \geq \max \{\max_{j \in J'} p_j, \frac1m \sum_{j\in J'}p_j \}.
\)
To apply the framework, we can use List Scheduling~\cite{DBLP:journals/siamam/Graham69} to compute non-preemptive schedule for any job set $J'$ with makespan at most
\(
\rho \cdot \max \{\max_{j \in J'} p_j, \frac{1}{m}\sum_{j\in J'}p_j \}
\)
for $\rho = 2$.
Hence, \Cref{thm: framework} implies a $(4e + \varepsilon)$-approximation algorithm for minimizing the total weighted group completion time on parallel identical machines. We note that a $2$-approximation is known for this problem~\cite{leung2006approximation,DBLP:journals/scheduling/YangP05}.

\medskip 
\noindent \textbf{Uniformly Related Machines.}
We are given $m$ parallel machines with speeds $s_1 \geq s_2 \geq \ldots \geq s_m$ and assume w.l.o.g.\ that $m=n$. %
At any time $t$, a feasible schedule satisfies for all $\{j_1, j_2\ldots, j_{\ell} \} \subseteq J$ that
$\sum_{q = 1}^{\ell} y_{j_q,t} \leq \sum_{i =1}^{\ell} s_i$
\cite{FeldmanMNP08,DBLP:journals/jacm/ImKM18}:
These are the packing constraints of our polytope $\mathcal P$. They can be separated in polynomial time by enumerating $\ell$ and checking the first $\ell$ jobs with largest $y_{j,t}$.
Let $J'=\{j_1, j_2\ldots, j_k \} \subseteq J$ be any subset of jobs with $p_{j_1} \geq \ldots \geq p_{j_k}$.
We have  %
\begin{align*}
  \max_{d\in [D]} \sum_{j\in J'}b_{d,j}p_j =
  \max_{d\in [D]} \sum_{q=1}^k b_{d,j_q}p_{j_q}
  \geq \max_{1 \leq \ell \leq k} \sum_{q=1}^{\ell} p_{j_q} \cdot \bigg(\sum_{i=1}^{\ell} s_i \bigg)^{-1} \ .
\end{align*}
The last inequality uses the packing constraints on the first $\ell$ jobs in $J'$ for $1 \leq \ell \leq k$.
The polynomial-time algorithm in \cite{DBLP:journals/jacm/HorvathLS77} computes a \emph{preemptive} schedule with a makespan of at most $\max_{1 \leq \ell \leq k} \sum_{j=1}^{\ell} p_j / \sum_{i=1}^{\ell}s_i$ for any $J'=\{j_1, j_2\ldots, j_k \} \subseteq J$.
Combined with a polynomial-time conversion from a \emph{preemptive} schedule to a \emph{non-preemptive} schedule, which increases the makespan by a factor of at most $(2-\frac{1}{m})$~\cite{DBLP:journals/orl/Woeginger00a}, 
we can compute in polynomial time a non-preemptive schedule with makespan
\[
     2 \cdot T \leq 2 \cdot \max_{1 \leq \ell \leq k} \sum_{q=1}^{\ell} p_{j_q} \cdot \bigg(\sum_{i=1}^{\ell} s_i \bigg)^{-1} \leq 2 \cdot \max_{d\in [D]} \sum_{j\in J'}b_{d,j}p_j \ ,
\]
which yields the desired subroutine.
Hence, \Cref{thm: framework} implies a $4e + \varepsilon \leq 10.874$-approximation for
\abc{Q}{r_j}{\sum w_SC_S}, which improves upon the known $13.5$-approximation for unrelated machines~\cite{DBLP:journals/mor/CorreaSV12}.
\begin{theorem}
  There is %
  a $10.874$-approximation algorithm for \abc{Q}{r_j}{\sum w_SC_S}.
\end{theorem}

\subsection{Applications to %
Sum Multicoloring}\label{sec:dpsp-graph-scheduling}

Next, we consider the non-preemptive sum multicoloring problem with group completion times (npSMC-G), which generalizes a collection of coloring and scheduling problems, e.g, non-preemptive sum multicoloring, data migration, and coflow scheduling.
We first model feasible rates as a polytope.
Let $\mc$ be the family of cliques of $G$.
At any time, the allocation of rates to jobs should satisfy for all cliques $C \in \mc$ that
$\sum_{v\in C}y_{v,t} \leq 1$.
Given a subset of jobs $V' \subseteq V$, let $G[V']$ be the subgraph induced by $V'$ and $\mc[V']$ be the family of cliques in $G[V']$.
By expressing the above constraints as~\eqref{IntervalLP:polytope},
we derive
\[
\max_{d\in [D]} \sum_{j\in J'}b_{d,j}p_j \geq \max_{C \in \mc} \sum_{v \in C \cap V'} p_v = \max_{C \in \mc[V']} \sum_{v \in C} p_v\ .
\]
The last term is the length of the largest clique in $G[V']$.
To apply our algorithmic framework, the key is to find an algorithm that computes a schedule for any given job set $V' \subseteq V$ with a makespan of $\rho \cdot \max_{C \in \mc[V']} \sum_{v \in C} p_v$
for some $\rho \geq 1$.
By plugging in different algorithms for the subproblem, we obtain approximation algorithms for various graph classes as follows.

\medskip 
\noindent \textbf{Interval Graphs.}
  Buchsbaum et al.~\cite{DBLP:journals/siamcomp/BuchsbaumKKRT04} considered the makespan minimization variant of multicoloring and obtained a $(2+\varepsilon)$-approximation against the lower bound of the largest clique size, which implies an algorithm with $\rho=2+\varepsilon$.
  Further, if jobs' processing times are all in $\Theta(1)$, they obtained a PTAS.
  Applying \Cref{thm: framework}, 
  we obtain the following result.

  \begin{restatable}{theorem}{thmIntervalGraphs}
      There is a $10.874$-approximation algorithm for npSMC-G on interval graphs.
      If jobs' processing times are all in $\Theta(1)$, the approximation ratio can be improved to $5.437$.
  \end{restatable}

  \medskip 
\noindent \textbf{Perfect Graphs and Unit Processing Times.}
  In perfect graphs, we can always find a vertex coloring in polynomial time (cf.\ Corollary 67.2c~\cite{schrijver2003combinatorial}) whose number of colors is exactly the size of the largest clique.
  This implies a makespan minimization algorithm with $\rho = 1$.
  Applying \Cref{thm: framework}, we obtain the following, which improves the $10.874$-approximation in~\cite{DBLP:conf/swat/DarbouyF24}.

  \begin{restatable}{theorem}{thmPerfectGraphs}\label{thm: perfect-graph}
    There is a $5.437$-approximation algorithm for SC-G on perfect graphs.
  \end{restatable}

  \medskip 
\noindent \textbf{Line Graphs.} For simplicity, we flip the roles of vertices and edges.
We say an edge is a job and no two incident edges of the same vertex can be scheduled at the same time.
Hence, the length of the largest clique turns out to be the largest total incident edge lengths, i.e., $\max_{v \in V}\sum_{e\in \delta(v)} p_e$.
We remark that this is equivalent to graph scheduling.

We consider the following greedy algorithm, which turns out to be a makespan minimization algorithm with $\rho =2$:
Iterate through time $t$ and greedily schedule a job (edge) $e=(u,v)$ in the time interval $[t,t+p_e]$ whenever none of other incident edges of $u$ or $v$ is being scheduled at time $t$ (consider the jobs in an arbitrary order).
For the analysis, consider any edge $e=(u,v) \in E$.
Before the start time $t$ of $e$, either some incident edge of $u$ or some incident edge of $v$ is scheduled as otherwise, we would have scheduled $e$ earlier.
Hence, $C_e = t + p_e \leq p(\delta(u) \setminus \{e\})+p(\delta(v) \setminus \{e\}) + p_e \leq 2 \max_{v \in V}\sum_{e\in \delta(v)} p_e$.

Applying \Cref{thm: framework}, we obtain the following, which also applies to the two special cases, i.e., non-preemptive data migration and coflow scheduling with arbitrary processing times.

\begin{restatable}{theorem}{thmLineGraphs}\label{thm: line-graph}
There is a $10.874$-approximation algorithm for graph scheduling (or equivalently, npSMC-G on line graphs).
\end{restatable}

We remark that for the above graph classes, the constraints $\sum_{v\in C}y_{v,t} \leq 1$ for all $C \in \mc$ can be separated in polynomial time (cf.\ Theorem 67.6 in~\cite{schrijver2003combinatorial} for interval graphs and perfect graphs and~\cite{DBLP:journals/jal/Kim05} for line graphs).

\section{A $(2+\varepsilon)$-approximation for offline PSP-G} \label{sec:preemptive}

Finally, we give a randomized $(2+\varepsilon)$-approximation for offline PSP-G for any $\varepsilon > 0$. %
The idea originates from algorithms for PSP and the open shop problem~\cite{JLM25,DBLP:journals/jal/QueyranneS02,SchulzS97}.

Let $\delta>0$ and $\varepsilon' >0$.
We first compute an optimal solution $\{\xp_{j,i},\xp_{S,i}, \cp_S\}$ to~\eqref{IntervalLP}, which is a relaxation of PSP-G.
As discussed in \Cref{sec: DPSP}, we have $\sum_{S \in \ms}w_S\cp_S \leq (1+\delta)\opt$.
Note that the solution $\{\xp_{j,i},\xp_{S,i}, \cp_S\}$ implies a feasible preemptive schedule, which we call the LP schedule.
We will derive another schedule from the LP schedule whose objective value is at most $2(1+\varepsilon')\sum_{S\in \ms} w_S \cp_S \leq (2+\varepsilon')(1+\delta)\opt$.

Let $\alpha$ be a random real variable drawn from the interval $(0,1)$ with the density function $f(\theta)=2\theta$.
We obtain a new schedule by slowing down the LP schedule by a factor of $\frac{1}{\alpha}$.
That is, we map and stretch the allocation of the LP schedule in time interval $[t,t+1]$ to time interval $[\frac{t}{\alpha},\frac{t+1}{\alpha}]$.
Now each job $j$ receives $\frac{p_j}{\alpha}$ amount of scheduling in the stretched schedule.
Then we modify the rate of every job $j$ to $0$ after the time when $j$ has been completed.
Clearly, the final schedule is feasible.
The advantage of this schedule is that we can bound the (expected) completion time of any job $j$ (resp.\ group $S$) against $\cp_j$ (resp.~$\cp_S$).

We introduce some notation for the analysis.
For any job $j$, let $C_j^\alpha$ be the earliest time by which an $\alpha$-fraction of $j$ is completed in the LP schedule.
Fix some $\alpha \in (0, 1)$ and let $C_j$ (resp. $C_S$) be the completion time of $j$ (resp. $S$) in the final schedule of our algorithm. We first bound $C_j$ against the $\alpha$-point of $j$ in the LP schedule.

\begin{restatable}{lemma}{preemptiveJobTime}
For every $j \in J$, it holds $C_j \leq C_j^\alpha / \alpha$.
\end{restatable}
\begin{proof}
  In the LP schedule, an $\alpha$-fraction of $j$ has been done by the time $C_j^\alpha$.
  Since we slow down the LP schedule by a factor of $\frac{1}{\alpha}$, $j$ is completed by time $C_j^\alpha / \alpha$ in the final schedule.  
\end{proof}

Next, we bound the completion time of a group $S \in \ms$ in the final schedule.
We do this analogously to $C_j^\alpha$.
Consider the earliest interval $i$ by which an $\alpha$-fraction of $S$ has been completed, i.e., $\sum_{i'=1}^i x_{S,i'} \geq \alpha \text{ and } \sum_{i'=1}^{i-1} x_{S,i'} < \alpha$.
Let $C_S^\alpha$ be the leftmost point of the $i$th interval, i.e., $\gamma_{i-1}$.
We obtain the following.
\begin{restatable}{lemma}{preemptiveGroupTime}
  For any $S \in \ms$, $C_S = \max_{j \in S} C_j \leq (1+\varepsilon')C_S^\alpha /\alpha$.
\end{restatable}

\begin{proof}
  Consider the earliest interval $i$ by which an $\alpha$-fraction of $S$ has been completed, i.e.,
  \[
  \sum_{i'=1}^i \xp_{S,i'} \geq \alpha \text{ and } \sum_{i'=1}^{i-1} \xp_{S,i'} < \alpha \ .
  \]
  By~\eqref{IntervalLP:prefix-bound}, for any job $j \in S$,
  \[
    \sum_{i' =1}^{i} \xp_{S,i'} \leq \sum_{i'=i}^{i} \frac{\xp_{j,i'}}{p_j}|I_{i'}| \ ,
  \]
  which implies that an $\alpha$-fraction of $j$ has been done at time $\gamma_i$ (the end of $I_i$) and $C_j^\alpha \leq \gamma_i$.
  Hence, we have
  \[
    C_j^\alpha \leq \gamma_i = (1+\varepsilon')\gamma_{i-1} = (1+\varepsilon')C_S^\alpha \ .
  \]
  Further, we have
  \[
     C_S = \max_{j \in S} C_j \leq \max_{j \in S} \frac{C_j^\alpha}{\alpha} \leq (1+\varepsilon') \frac{C_S^\alpha}{\alpha} \ ,
  \]
  which concludes the lemma.
  \end{proof}

Finally, we bound the algorithm's performance %
over the random choice of $\alpha$.

\thmPreemptive*
\begin{proof}
Recall that $\alpha$ is randomly drawn from the interval $(0,1)$ with the density function $f(\theta)=2\theta$.
For any set $S$, we have
\[
  \EX_\alpha[C_S] \leq \int_0^1 (1+\varepsilon')\frac{C_S^\alpha}{\alpha} 2\alpha \; \mathrm{d}\alpha = 2(1+\varepsilon')\int_0^1 C_S^\alpha \; \mathrm{d}\alpha \ .
\]
Note that $C_S^\alpha$ is a non-decreasing, piecewise constant function of $\alpha$. Thus,
\[
        \int_0^1 C_S^\alpha \; \mathrm{d}\alpha = \sum_{1  \leq i \leq L}|\{\alpha: C_S^\alpha = \gamma_{i-1}\}| = \sum_{1  \leq i \leq L} x_{S,i} \cdot \gamma_{i-1} = \cp_S \ .
\]
The last equality follows from~\eqref{IntervalLP:group-finishes} and~\eqref{IntervalLP:CS}.
It follows that $\EX_\alpha[C_S] \leq 2(1+\varepsilon') \cp_S$, and we conclude with
\[
\EX_\alpha \bigg[\sum_{S\in \ms}w_S C_S \bigg] \leq  \sum_{S\in \ms}w_S\EX_\alpha [C_S] \leq 2(1+\varepsilon')\sum_{S\in \ms}w_S\cp_S \leq (2+\delta)(1+\varepsilon') \opt \ .
\]
This completes the proof of the theorem.
\end{proof}

Our algorithm is almost best-possible, since PSP-G contains \abc{1}{\pmtn}{\sum_{S}w_S C_S}, which has been shown to be equivalent to \abc{1}{\mathrm{prec}}{\sum_j w_j C_j}~\cite{DBLP:journals/jal/QueyranneS02} and has a lower bound of $2-\varepsilon$ under a stronger version of the unique game conjecture~\cite{DBLP:conf/focs/BansalK09}.

\section{Conclusion}
With polytope scheduling for group completion times, PSP-G and DPSP-G, we propose a unifying abstraction of several scheduling, coloring, and other graph problems. We present general algorithmic tools that achieve new and improved approximation factors.
While we give a meta-framework for solving the new Discrete PSP, particularly for the general DPSP-G (\Cref{sec: DPSP}), its implementation relies on problem-specific makespan minimization subroutines. It remains as an open question whether the dependence on problem-specific properties can be avoided. Further, for offline PSP-G, it is open whether the $(2+\varepsilon)$-approximation (\cref{sec:preemptive}) can be improved, even in the special case without groups.

\printbibliography

\end{document}